\documentclass[12pt]{article}
\usepackage{amsmath}
\usepackage{amssymb} 
\usepackage[utf8]{inputenc}
\usepackage{etoolbox,fp}
\usepackage[dvipsnames,usenames]{xcolor}
\usepackage{tikz}
\usetikzlibrary{shapes}
\newtheorem{theorem}{Theorem}[section]

\newtheorem{lemma}[theorem]{Lemma}

\newtheorem{corollary}[theorem]{Corollary}
\newtheorem{remark}[theorem]{Remark}
\newif\ifnotesw\noteswtrue
\newcommand{\comm}[1]{\ifnotesw $\blacktriangleright${\sf #1}$\blacktriangleleft$ \fi}
\noteswfalse	
\newcommand{\bull}{\mbox{$\;\;\;$\vrule height .9ex width .8ex depth -.1ex}}
\newenvironment{proof}{\par\smallbreak\noindent{\bf Proof.~}}
{\unskip\nobreak\hfill \bull \par\medbreak}
\newcounter{clm}
\renewcommand{\theclm}{\Alph{clm}}
\newenvironment{clm}{\refstepcounter{clm}%
\par\medskip\par\noindent{\it Claim~\theclm.~}~\rm}%
{\par\smallskip\par}
\newenvironment{subproof}{\par\noindent{\sl Proof of Claim~\theclm.~}}%
{$\,\triangleleft$\par\medskip\par}

\newcommand{\hide}[1]{}
\newcommand{\refeq}[1]{(\ref{eq:#1})}
\newcommand{\setdef}[2]{\left\{ \hspace{0.5mm} #1 : \hspace{0.5mm} #2 \right\}}
\newcommand{\bbz}{\mathbb{Z}}

\newcommand{\e}[1]{||#1||}
\newcommand{\dd}[1]{\mathit{depth}(#1)}
\newcommand{\len}[1]{\mathit{length}(#1)}
\newcommand{\lenast}[1]{\mathit{length}^\ast(#1)}
\newcommand{\s}[1]{\mathit{size}(#1)}
\newcommand{\Time}[1]{\mathit{time}(#1)}

\newcommand{\game}{\mbox{\sc Game}}
\newcommand{\cert}{\mathrm{Cert}}
\newcommand{\cnt}{c}
\newcommand{\ACgraph}{{AC'13}}
\newcommand{\DOMINO}{{\small\scshape Domino}}
\newcommand{\COWHEELS}{{\small\scshape Co-Wheels}}

	\newcommand{\GAME}{{\small\scshape AC-Problem}}
	\newcommand{\Hornsat}{{\scshape Horn-Sat}}
	\usepackage{algorithm}
	\usepackage{algpseudocode}
	\algnotext{EndFor}
	\algnotext{EndIf}
			\newcommand{\problembox}[3]{\smallskip\par
			\begin{center}
				\fbox{\parbox{.9\columnwidth}{
					#1\par \vspace{2mm}
					\begin{tabular}{rl}
							\textit{Input}: & #2\\
							\textit{Question}: & #3 
					\end{tabular}
				}}
			\end{center}\smallskip
		}

\colorlet{acolor}{Apricot!30}
\colorlet{bcolor}{blue}
\colorlet{ccolor}{cyan}
\colorlet{dcolor}{Dandelion}

\tikzset{
a/.style={->},
aa/.style={<-},
aaa/.style={->},
aarrow/.style={draw=blue, ultra thick,<-},
barrow/.style={draw=Red!70, ultra thick,<-},
  avertex/.style={circle,draw,inner sep=2pt,fill=acolor},
  aavertex/.style={circle,draw,inner sep=2.5pt,fill=Apricot},
bvertex/.style={circle,draw,inner sep=2pt,fill=bcolor},
rvertex/.style={regular polygon,regular polygon sides=6,draw,inner sep=2pt,fill=Red},
cvertex/.style={circle,draw,inner sep=2pt,fill=ccolor},
dvertex/.style={regular polygon,regular polygon sides=6,draw,inner sep=2.1pt,fill=dcolor},
  uvertex/.style={circle,draw=Black,inner sep=2pt,fill=white},
  schvertex/.style={regular polygon,regular polygon sides=6,draw,inner sep=2pt,fill=Red},
  vertex/.style={circle,draw=Gray,inner sep=2pt,fill=Gray}
}

\title{On the speed of constraint propagation and\\ 
the time complexity of arc consistency testing}

\author{
Christoph Berkholz \\[2mm]
\normalsize
RWTH Aachen University \\
\normalsize
Institut f\"ur Informatik \\
\normalsize
D-52056 Aachen
\and
Oleg Verbitsky\thanks{%
Supported by DFG grant VE 652/1--1.
On leave from the Institute for Applied Problems of Mechanics and Mathematics,
Lviv, Ukraine.}\\[2mm]
\normalsize
Humboldt-Universit\"at zu Berlin\\
\normalsize
Institut f\"ur Informatik\\ 
\normalsize
Unter den Linden 6\\
\normalsize
D-10099 Berlin
}

\date{}

\begin{document}

\maketitle

\begin{abstract}
Establishing arc consistency on two relational structures
is one of the most popular heuristics for the constraint satisfaction problem.
We aim at determining the time complexity of arc consistency testing.
The input structures $G$ and $H$ can be supposed to be connected
colored graphs, as the general problem reduces to this particular case.
We first observe the upper bound $O(e(G)v(H)+v(G)e(H))$, which implies
the bound $O(e(G)e(H))$ in terms of the number of edges and
the bound $O((v(G)+v(H))^3)$ in terms of the number of vertices. 
We then show that both bounds are tight
up to a constant factor as long as an arc consistency algorithm is based
on constraint propagation (like any algorithm currently known).

Our argument for the lower bounds is based on examples of slow constraint
propagation. We measure the speed of constraint propagation
observed on a pair $G,H$ by the size of a proof, in a natural combinatorial proof system,
that Spoiler wins the existential 2-pebble game on $G,H$.
The proof size is bounded from below by the game length $D(G,H)$, and
a crucial ingredient of our analysis is the existence of $G,H$
with $D(G,H)=\Omega(v(G)v(H))$. 
We find one such example among old benchmark instances
for the arc consistency problem and also suggest a new, different construction.
\end{abstract}

\section{Introduction}
According to the framework of \cite{Feder.1998}, the \emph{constraint satisfaction problem} (CSP) takes two finite relational structures as input and asks whether there is a homomorphism between these structures. In this paper we consider structures with unary and binary relations and refer to unary relations as colors and to binary relations as directed edges. In fact, most of the time we deal with structures having only one binary, symmetric and irreflexive relation $E$, i.e., with vertex-colored graphs. This is justified by a linear time reduction from the CSP on binary structures to its restriction on colored graphs;
see Section \ref{sec:reduction}. 
Note that the CSP restricted to colored graphs and digraphs has also been studied under the name ``List Homomorphism'' from an algebraic point of view.

Let $G$ and $H$ be an input of the CSP. It is customary to call the vertices of $G$ \emph{variables} and the vertices of $H$ \emph{values}. A mapping from $V(G)$ to $V(H)$ then corresponds to an assignment of values to the variables, and the assignment is \emph{satisfying} if the mapping defines a homomorphism. 
Let a \emph{domain} $D_x\subseteq V(H)$ of a variable $x\in V(G)$ be a set of admissible assignments to this variable. Formally, $D_x$ is a domain if for every homomorphism $h: G\to H$ it holds that $h(x)\in D_x$. 
The aim of the arc consistency heuristic is to find small domains in order to shrink the search space. 
The first step of the arc consistency approach is to ensure \emph{node consistency}, that is, $D_x$ is initialized to the set of vertices in $H$ that are colored with the same color as $x$. 
The second step is to iteratively shrink the domains according to the following rule: 

\begin{quote}\label{quote:AC-rule}
If there exists an $a\in D_x$ and a variable $y\in V(G)$ such that $\{x,y\}\in E(G)$ and $\{a,b\}\notin E(H)$ for all $b\in D_y$, then delete $a$ from $D_x$. 
\end{quote}
 
A pair of graphs augmented with a set of domains
is \emph{arc consistent} if the above rule cannot be applied and all domains are nonempty. We say that arc consistency \emph{can be established} for $G$ and $H$, if there exists a set of domains such that $G$ and $H$ augmented with these domains is arc consistent. 
Our aim is to estimate the complexity of the following decision problem.
\problembox{\GAME{}}{Two colored graphs $G$ and $H$.}{Can arc consistency be established on $G$ and $H$?}
Using known techniques for designing arc consistency algorithms, we observe that the \GAME{} can be solved in time $O(v(G)e(H)+e(G)v(H))$, where $v(G)$ and $e(G)$ denote the number of vertices and
the number of edges respectively. 
Since this gives us only a quadratic upper bound in terms of the overall input size, there could be a chance for improvement: Is it possible to solve the \GAME{} in sub-quadratic or even linear time? In fact, we cannot rule out this possibility completely. The first author \cite{Berkholz.2012} recently obtained lower bounds for higher levels of $k$-consistency (note that arc consistency is equivalent to 2-consistency). In particular, 15-consistency cannot be established in linear time and establishing 27-consistency requires more than quadratic time on multi-tape Turing machines. The lower bounds are obtained in \cite{Berkholz.2012} via the deterministic time hierarchy theorem and, unfortunately, these methods are not applicable to arc consistency because of the blow-up in the reduction. 

However, we show lower bounds for every algorithm that is based on constraint propagation. A \emph{propagation-based arc consistency algorithm} is an algorithm that solves the \GAME{} by iteratively shrinking the domains via the arc consistency rule above. Note that all currently known arc consistency algorithms 
(as e.g. AC-1, AC-3 \cite{Mackworth.1977}; AC-3.1/AC-2001 \cite{Bessiere.2006}; AC-3.2, AC-3.3 \cite{Lecoutre.2003};
AC-3$_d$ \cite{Dongen.2002}; 
AC-4 \cite{Mohr.1986}; AC-5 \cite{VanHentenryck.1992}; 
AC-6 \cite{Bessiere.1994}; 
AC-7 \cite{Bessiere.1999}; 
AC-8 \cite{Chmeiss.1998}, 
AC-$\ast$ \cite{Regin.2005})
are propagation-based in this sense. Different AC algorithms differ in the principle of ordering propagation steps;
for a general overview we refer the reader to \cite{Bessiere.2006}. 
The upper bound $O(v(G)e(H)+e(G)v(H))$ implies $O(e(G)e(H))$ in terms of the number of edges and $O(n^3)$ in terms of the number of vertices $n=v(G)+v(H)$. Our main result, Theorem \ref{thm:time-lower} in Section \ref{sec:time}, states that both bounds are tight up to a constant factor for any propagation-based algorithm.

We obtain the lower bounds by exploring a connection between the \emph{existential 2-pebble game} and propagation-based arc consistency algorithms.
In its general form the existential $k$-pebble game is an Ehrenfeucht-Fra\"iss\'e like game that determines whether two finite structures can be distinguished in the existential-positive $k$-variable fragment of first order logic. It has found applications also outside of finite model theory: to study the complexity and expressive power of Datalog \cite{Kolaitis.1995}, $k$-consistency tests \cite{Kolaitis.2000a,Kolaitis.2003,Atserias.2007,Berkholz.2012} and bounded-width resolution \cite{Atserias.2008,Berkholz.2012a}. 
It turns out that the existential 2-pebble game exactly characterizes the power of arc consistency \cite{Kolaitis.2000a}, i.e., Spoiler wins the existential 2-pebble game on two colored graphs $G$ and $H$ iff arc consistency cannot be established.

The connection between the existential 2-pebble game and arc consistency algorithms is deeper than just a reformulation of the \GAME{}. We show that every constraint propagation-based arc consistency algorithm computes in by-passing a proof of Spoiler's win on instances where arc consistency cannot be established. 
On the one hand these proofs of Spoiler's win naturally correspond to a winning strategy for Spoiler in the game. On the other hand they reflect the propagation steps performed by an algorithm. 
We consider three parameters to estimate the complexity of such proofs: length, size and depth. The length corresponds to the number of propagation steps, whereas size also takes the cost of propagation into account. The depth corresponds to the number of ``nested'' propagation steps and precisely matches the number of rounds $D(G,H)$ Spoiler needs to win the game. We observe that the minimum size of a proof of Spoiler's win on $G$ and $H$ bounds from below the running time of sequential propagation-based algorithms, whereas the minimal depth matches the running time of parallel algorithms.

We exhibit pairs of colored graphs $G,H$ where $D(G,H)=\Omega(v(G)v(H))$ and hence many nested propagation steps are required to detect arc-inconsistency. Because these graphs have a linear number of edges this implies that there is no sub-quadratic  propagation-based arc consistency algorithm. 
It should be noted that CSP instances that are hard for sequential and parallel arc consistency algorithms, in the sense that they require many propagation steps, have been explored very early in the AI-community \cite{Dechter.1985,Samal.1987}. Such examples were also proposed to serve as benchmark instances to compare different arc consistency algorithms \cite{Bessiere.2005}.
Graphs $G$ and $H$ with large $D(G,H)$ can be derived from the old \DOMINO\  example,
consisting of structures with two binary relations. We also provide a new example,
which we call \COWHEELS, that shows the same phenomenon of slow constraint propagation
for a more restricted class of rooted loopless digraphs.
\comm{The root can be modelled by a loop in the second binary relation.}

The rest of the paper is organized as follows. In Section \ref{sec:prelim} we give
the necessary information on the existential 2-pebble game and use it to analyze
the \DOMINO\ pattern. Our \COWHEELS\ pattern is introduced and analyzed in Section \ref{sec:games}.
Section \ref{sec:proofs} is devoted to the winner proof system for the  existential 2-pebble game.
The facts obtained here are used in Section \ref{sec:time}
to prove our main results on the complexity of 
propagation-based algorithms for the \GAME{}.

\section{Preliminaries}
\label{sec:prelim}

A \emph{binary structure} $A$ is a relational structure
over vocabulary $\sigma=\{E_1,E_2,\ldots,U_1,\allowbreak U_2,\ldots\}$ consisting of binary relations $E_i$, $i\ge1$,
and unary relations $U_j$, $j\ge1$. Each binary relation $E_i^A$
between elements of $A$ can be regarded as a directed graph 
with arrows $(x,y)\in E_i^A$ colored in color $i$.
Similarly, the unary relations $U_j^A$ can be thought of as colors
of elements of $A$. In this way, we can consider $A$ an edge- and vertex-colored
directed graph. The elements of $A$ will be then called \emph{vertices}.
The set of the elements of $A$ will be denoted by $V(A)$ and their number by $v(A)$.

In \emph{colored graphs} we additionally have unary relations of
vertex colors, i.e., $\sigma=(E_1,U_1,U_2,\ldots)$.
Moreover, it is supposed that any two color classes $U^A_j$ and $U^A_{j'}$
are disjoint.
The number of edges in a colored graph $A$ is denoted by $e(A)$.

The \emph{existential $2$-pebble game on binary structures $A$ and $B$}
is played by two players, Spoiler and Duplicator,
to whom we will refer as he and she respectively.
The players have equal sets of two pairwise different pebbles, $p$ and $q$.
A \emph{round} consists of a move of Spoiler followed by a move of
Duplicator. Spoiler takes a pebble, $p$ or $q$, and puts it on a vertex in $A$.
Then Duplicator has to put her copy of this pebble on a vertex 
of $B$. Duplicator's objective is to keep the following
condition true after each round: the pebbling should determine a partial
homomorphism from $A$ to~$B$.

Let $x\in V(A)$ and $u\in V(B)$ denote the vertices pebbled by $p$
and $y\in V(A)$ and $v\in V(B)$ denote the vertices pebbled by $q$.
Thus, Duplicator loses as soon as $x\in U_j^A$ while $u\notin U_j^B$ 
for some $j$, or $(x,y)\in E_i^A$ while $(u,v)\notin E_i^B$,
 or $(y,x)\in E_i^A$ while $(v,u)\notin E_i^B$ for some $i$, or
$x=y$ while $u\ne v$.

For each positive integer $r$, the $r$-round $2$-pebble existential game
on $A$ and $B$
is a two-person game of perfect information with a finite number of positions.
By Zermelo's theorem, either Spoiler or Duplicator has a \emph{winning strategy}
in this game, that is, a strategy winning against every strategy of the opponent.
Let $D(A,B)$ denote the minimum $r$ for which Spoiler has a winning strategy.
If such $r$ does not exist, we will write $D(A,B)=\infty$.
As it is well known \cite{Kolaitis.1995},
$D(A,B)\le r$ if and only if $A$ can be distinguished from $B$
by a sentence of quantifier rank $r$ in the existential-positive two-variable logic.
The \emph{existential-positive} fragment of first-order logic consists of formulas
containing only monotone Boolean connectives and only existential quantifiers
(thus, negation and universal quantification is forbidden).

Suppose that $D(A,B)<\infty$. We say that Spoiler plays \emph{optimally} if
he never loses an opportunity to win as soon as possible. More specifically, after a round
is ended in a position $P$ (determined by the pebbled vertices), Spoiler
makes the next move according a strategy that allows him to win from the position $P$
in the smallest possible number of rounds.

\begin{lemma}\label{lem:optimal-str}
If Spoiler plays optimally, then the following
conditions are true.
\begin{enumerate}
\item 
Spoiler uses the pebbles alternatingly, say, $p$ in odd and $q$ in even rounds.
\item 
Whenever Spoiler moves a pebble, he moves it to a new position.
That is, if $x_i\in V(A)$ denotes the vertex pebbled in the $i$-th round, then
$x_{i+2}\ne x_i$.
\comm{Also $x_{i+1}\ne x_i$ if $A$ is loopless.}
Moreover, if $x_{i+1}=x_i$, then $x_{i+2}\ne x_{i-1}$.
\item 
 $(x_i,x_{i+1})$ or $(x_{i+1},x_i)$
satisfies at least one binary relation.
\end{enumerate}
\end{lemma}

\begin{proof}
Recall that a position $P$ in the game is a tuple in $V(A)^2\times V(B)^2$ or in $V(A)\times V(B)$
consisting of the currently pebbled vertices.
By assumption, Spoiler has a strategy allowing him to win the game within some number of rounds.
Then, for every $P$ there is an $r$ such that Spoiler has a winning strategy in the $r$-round
game with the initial position $P$. Denote the smallest such $r$ by $R(P)$.
We will denote the vertex of $B$ pebbled by Duplicator in the $i$-th round by~$u_i$.

  {\bf 1.}
We first show this for the first two rounds. Let $R$ denote the minimum number $r$
such that Spoiler has a winning strategy in the $r$-round game.
It is clear that in the first round Spoiler pebbles a vertex $x_1\in V(A)$ such
that $\max_{u\in V(B)}R(x_1,u)$ is equal to the minimum possible value $R-1$.
If in the second round Spoiler just moves the pebble from $x_1$
to another vertex $x_2$, then Duplicator can
pebble a vertex $u_2$ attaining $\max_{u\in V(B)}R(x_2,u)\ge R-1$.
This allows her to win the next $r-2$ rounds, contradictory to the
fact that the optimal strategy used by Spoiler is winning in the
$R$-round game.

Assume now that Spoiler has used the pebble $p$ in the $(i-1)$-th round
and the pebble $q$ in the $i$-th round, and the game is not over yet. 
By the definition of an optimal
strategy, the value $R'=\max_{u}R(x_{i-1},x_i,u_{i-1},u)$ is minimum
possible among all choices of $x_i$. From now on Spoiler has to win
the game in at most $r'$ rounds. If, however, in the $(i+1)$-th round 
Spoiler uses the pebble $q$ again moving it from $x_i$
to $x_{i+1}$, then Duplicator can
pebble a vertex $u_{i+1}$ attaining $\max_{u}R(x_{i-1},x_{i+1},u_{i-1},u)\ge R'$.
This allows her to win the further $R'-1$ rounds, contradicting Spoiler's optimality.

  {\bf 2.}
The definition of an optimal strategy implies that, after the $i$-th round is played,
Spoiler wins in at most $R(x_{i-1},x_i,u_{i-1},u_i)$ rounds. Assume that
in the $(i+1)$-th round Spoiler pebbles $x_{i+1}=x_{i-1}$. Not to lose immediately,
Duplicator pebbles $u_{i+1}=u_{i-1}$. Starting from the next round, Duplicator
is able to stand up in $R(x_i,x_{i-1},u_i,u_{i-1})-1=R(x_{i-1},x_i,u_{i-1},u_i)-1$
rounds, which gives a contradiction.

If $x_{i+1}=x_i$, the inequality $x_{i+2}\ne x_{i-1}$ follows by a similar argument.

{\bf 3.}
Part 1 of the lemma shows that after the $i$-th round the players actually play
the game with the initial position $(x_i,u_i)$ (that is, Spoiler's optimal
strategy can be supposed to be independent of the pair $(x_{i-1},u_{i-1})$).
In particular, Spoiler has a strategy allowing him to win the rest of the game
in $R(x_i,u_i)\le R-i$ rounds, where $R$ is as defined above.
Assume that the vertices $x_i$ and $x_{i+1}$ satisfy no binary relation in $A$.
Then every choice of $u_{i+1}\in V(B)$ is non-losing for Duplicator in the $(i+1)$-th round.
If she chooses $u_{i+1}$ attaining $\max_{u\in V(B)}R(x_{i+1},u)\ge R-1$,
then she has a strategy allowing her to survive at  least $R-1$ further rounds
after the $i$-th round, a contradiction.
\end{proof}

Lemma \ref{lem:optimal-str} has several useful consequences.
The first of them implies that, without loss of generality,
we can restrict our attention to connected structures.
Two distinct vertices of a binary structure $A$ are adjacent
in its \emph{underlying graph} $G_A$ if they satisfy at least
one binary relation of $A$. \emph{Connected components} of $A$
are considered with respect to $G_A$.
  Let $A$ consist of connected components $A_1,\ldots,A_k$ and
$B$ consist of connected components $B_1,\ldots,B_l$. Then it easily
follows from part 3 of Lemma \ref{lem:optimal-str} that
$
D(A,B)=\min_i\max_j D(A_i,B_j)
$.
\comm{
\begin{proof}
Suppose that $A$ is distinguishable from $B$ in the two-variable existential-positive logic.
Then $D(A,B)$ is equal to the minimum $r$ such that 
Spoiler has an optimal winning strategy in the existential 2-pebble game on $A$ and $B$.
If in the first round Spoiler pebbles a vertex in $A_i$ and Duplicator in $B_j$,
then part 3 of Lemma \ref{lem:optimal-str} shows that, as long as Duplicator is alive,
the game is played on $A_i$ and $B_j$. This implies the claimed equality.

It remains to argue that, if $D(A,B)=\infty$, then for every $i$ there is $j$ such that
$D(A_i,B_j)=\infty$. Indeed, otherwise for some $i$ Spoiler would have a winning strategy 
in the game on $A_i$ and $B_j$ for each $j$ and, playing optimally, he then could
win the game on $A$ and $B$ because, again by part 3 of  Lemma~\ref{lem:optimal-str},
the game would be restricted to $A_i$ and the component $B_j$ determined by the first
move of Duplicator.
\end{proof}
}
Another consequence follows from parts 2 and~3.

\begin{corollary}
  \label{cor:A-tree}
Suppose that the underlying graph $G_A$ of $A$ is a tree.
  If $D(A,B)<\infty$, then $D(A,B)<2\,v(A)$.
\end{corollary}

\begin{proof}
Consider the existential 2-pebble game on $A$ and $B$ and assume that
Spoiler follows an optimal strategy.
 By part 3 of  Lemma \ref{lem:optimal-str}, he all the time moves the pebbles
along a path in $G_A$. By part 2 of the lemma, he never turns back.
Since $G_A$ is a tree, the game lasts at most $2\,d(G_A)+1<2\,v(A)$ rounds,
where $d(G_A)$ denotes the diameter of~$G_A$.
\comm{Is this bound optimal? Can the diameter be replaced with the radius?
Is there a similar bound for the full two-variable existential logic?}
\end{proof}

Furhtermore, we now can state a general upper bound for~$D(A,B)$.

\begin{corollary}
  \label{cor:upper}
If $D(A,B)<\infty$,
then $D(A,B)\le v(A)v(B)+1$.
\end{corollary}

\begin{proof}
Assume that Spoiler plays optimally.
Let $x_i\in V(A)$ and $u_i\in V(B)$ denote the vertices pebbled in the $i$-th round.
 By part 1 of  Lemma \ref{lem:optimal-str}, we can further assume that
Spoiler's move in the $(i+1)$-th round depends only on the $(x_i,u_i)$.
It readily follows that, if the game lasts $r$ rounds, then the pairs
$(x_1,u_1),\ldots,(x_{r-1},u_{r-1})$ are pairwise different, and hence
$r-1\le v(A)v(B)$.
\comm{The $+1$ term in the bound can be removed if there are pairs $(x,u)$
where Duplicator is dead immediately because $x$ has some color while $u$
doesn't. Can it be removed in general, for example, if $A$ and $B$
are digraphs with 2 colors for edges and no colors for vertices?}
\end{proof}

\begin{figure}
\centering
 \begin{tikzpicture}
  \FPset{\r}{2}
\begin{scope}[scale=0.5]
  \node at (-\r-3 cm,\r+1 cm) {$A_5$};
  \foreach \i in {0,...,4}
   {
    \node[uvertex] (x\i) at (\i*72:\r cm) {};
    \FPsub{\j}{\i}{1}
    \FPround{\j}{\j}{0}
    \ifnumgreater{\i}{0} 
{\draw[barrow] (x\j) -- (x\i);}{}
    }
    \draw[barrow] (x4) -- (x0);
    \draw[aarrow] (x2) -- (x3);
\end{scope}

\begin{scope}[xshift=40mm,scale=0.7]
  \node at (-1cm,\r+.7 cm) {$B_7$};
  \foreach \i in {0,...,6}
   {
    \node[uvertex] (u\i) at (\i cm, 0 cm) {};
    \path (u\i) edge[
out=120, in=60, looseness=0.8, loop, distance=1cm, ->, ultra thick, draw=Red!70] (u\i);
    \FPsub{\j}{\i}{1}
    \FPround{\j}{\j}{0}
    \ifnumgreater{\i}{0} 
      {\draw[aarrow] (u\i) -- (u\j);
      }{}
    }
 \end{scope}
 \end{tikzpicture}
\caption{The \DOMINO\  example.}
\label{fig:domino}
\end{figure}

The bound of Corollary \ref{cor:upper} is tight, at least, up to a factor of $1/2$.
A suitable lower bound can be obtained from the CSP instances that appeared
in \cite{Dechter.1985,Bessiere.2005}
under the name of \emph{DOMINO problem} and were used for benchmarking
the arc consistency algorithms.
A \DOMINO\  instance consists of two digraphs $A_m$ and $B_n$ whose arrows are colored
in red and blue; see Fig.~\ref{fig:domino}. $A_m$ is a directed cycle of length $m$ with
one blue and the other red arrows. $B_n$ is a blue directed path where red loops
are attached to all its $n$ vertices. Spoiler can win the existential 2-pebble game
on $A_m$ and $B_n$ by moving the pebbles along the cycle $A_m$, always in the same
direction. By Lemma \ref{lem:optimal-str}, this is the only way for him to win
in the minimum number of rounds. When Spoiler passes red edges, Duplicator
stays with both pebbles at the same vertex of $B_n$. Only when Spoiler passes the blue edge,
Duplicator passes one (blue) edge forward in $B_n$. Thus, if Duplicator
starts playing in the middle of $B_n$, she survives in at least $\frac12\,m(n-1)$
rounds.

\section{More examples of slow constraint propagation}
\label{sec:games}

The \DOMINO\  pairs are remarkable examples of binary structures on which
constraint propagation is as slow as possible, up to a constant
factor of $1/2$. An important role in the \DOMINO\  example is played
by the fact that we have two different edge colors.
We now show that essentially the same lower bound
holds true over a rather restricted class of structures, namely \emph{rooted
loopless digraphs}, where edges are uncolored,
there is a single color for vertices, and only a single \emph{root} vertex is colored in it.
It is also supposed that every vertex of a rooted digraph is reachable
from the root along a directed path.

\begin{figure}
\centering
 \begin{tikzpicture}[scale=0.5,>=stealth]
  \FPset{\r}{2}
\begin{scope}
  \node at (-\r cm,\r cm) {$G_4$};
  \node[rvertex] (r) at (0,0) {};
  \foreach \i in {0,...,3}
   {
    \node[uvertex] (x\i) at (\i*90:\r cm) {};
    \FPsub{\j}{\i}{1}
    \FPround{\j}{\j}{0}
    \ifnumgreater{\i}{0} 
{\draw[<-] (x\j) -- (x\i);}{}
    }
    \draw[<-] (x3) -- (x0);
    \draw[->] (r) -- (x0);
   \node[right] (egal) at (x0) {$\,\,x_0$};
\end{scope}

\begin{scope}[xshift=90mm]
  \node at (-\r cm,\r cm) {$H_5$};
   \node[rvertex] (r) at (0,0) {};
  \foreach \i in {0,...,4}
   {
    \node[uvertex] (x\i) at (\i*72:\r cm) {};
    \FPsub{\j}{\i}{1}
    \FPround{\j}{\j}{0}
    \ifnumgreater{\i}{0} 
{\draw[<-] (x\j) -- (x\i); \draw[->] (r) -- (x\i);}{}
    }
    \draw[<-] (x4) -- (x0);
   \node[right] (egal) at (x0) {$\,\,a_0$};
\end{scope}
 \end{tikzpicture}
\caption{An example of \COWHEELS\ .}
\label{fig:cowheels}
\end{figure}

By the \emph{wheel} $W_n$ we mean the rooted digraph with $n+1$ vertices
where there are arrows from the root to all the other $n$ vertices
and these vertices form a directed cycle. We call a pair of rooted digraphs
$G_m$ and $H_n$ \emph{co-wheels} if $G_m$ is obtained from $W_m$ by removal
of all but one arrows from the root and $H_n$ is obtained from $W_n$ by removal
of one arrow from the root; see an example in Fig.~\ref{fig:cowheels}.

\begin{lemma}
  \label{lem:cowheels}
Let $G_m$ and $H_n$ be co-wheels.
If $m$ and $n$ are coprime, then $D(G_m,\allowbreak H_n)<\infty$  and $D(G_m,H_n)>\frac12\,m(n-3)$.
\end{lemma}

\begin{proof}
  Let $V(G_m)=\{x_0,\ldots,x_{m-1}\}$ and $V(H_n)=\{a_0,\ldots,a_{n-1}\}$.
Assume that $x_0$ is adjacent to the root of $G_m$,
$a_0$ is non-adjacent to the root of $H_n$, and the indices increase
in the direction of arrows.
We first argue that Spoiler has a winning strategy in the existential 2-pebble
game on $G_m$ and $H_n$. Let Spoiler pebble $x_0$ in the first round
and assume that Duplicator responds with $a_t$. If $t=0$, Spoiler wins
by putting the other pebble on the root. If $t>0$, Spoiler is able to force
pebbling the pair $(x_0,a_0)$ in a number of rounds. Indeed, if Spoiler 
moves the pebbles alternatingly along the cycle so that the pebbled vertices are always
adjacent, then after $\ell m$ rounds Spoiler passes the cycle $\ell$ times
and arrives again at $x_0$, while Duplicator is forced to come to $a_{t+\ell m}$,
where the index is computed modulo $n$. Since $m$ and $n$ are coprime,
$m\bmod n$ is a generator of the cyclic group $\bbz_n$. It follows that
the parameter $\ell$ can be chosen so that
 $t+\ell m=0\pmod n$,
and then $a_{t+\ell m}=a_0$.

We now have to show that Duplicator is able to stand up in at least $\frac12\,m(n-3)$
rounds. Estimating the length of the game, we can assume that Spoiler plays
according to an optimal strategy. It readily follows by Lemma \ref{lem:optimal-str}
that Spoiler begins playing in a non-root vertex $x_s$ and forces pebbling the pair $(x_0,a_0)$
as explained above, by moving along the cycle always in the same direction.
Let $D(x_s,a_t)$ denote the minimum number of moves needed for Spoiler to reach
this configuration if Duplicator's move in the first round is~$a_t$.

Suppose first that $s=0$ and
also that Spoiler moves in the direction of arrows. Then he can force
pebbling $(x_0,a_0)$ only in $\ell m$ moves with $\ell$ satisfying $t+\ell m=0\pmod n$.
Denote $l=\lfloor n/2\rfloor$ and let Duplicator choose $t=(-lm)\bmod n$.
Then the smallest possible positive value of $\ell$
is equal to $l$. If Spoiler decides to move in the opposite direction,
we have the relation $t-\ell m=0\pmod n$, which gives us
$\ell\ge\lceil n/2\rceil$. In both cases $D(x_0,a_t)\ge\frac12\,m(n-1)$.

Suppose now that $s>0$. Let Duplicator pebble $a_{t'}$ in the first round with 
$t'=(t+s)\bmod n$, where $t$ is fixed as above. Note that Spoiler 
from the position $(x_0,a_t)$ is able to force the position $(x_s,a_{t'})$
in $s$ moves. Therefore, $D(x_0,a_t)\le s+D(x_s,a_{t'})$, which implies that
$D(x_s,a_{t'})\ge D(x_0,a_t)-(m-1)>\frac12\,m(n-3)$, as claimed.
\end{proof}

\begin{theorem}
  \label{thm:cowheels}
For every pair of numbers $M\ge5$ and $N\ge5$, there is a pair of
rooted loopless digraphs $G$ and $H$ with $v(G)=M$ and $v(H)=N$
such that $D(G,H)<\infty$ and $D(G,H)\ge(\frac12-o(1))MN$.
Here the $o(1)$-term is a function of $\max(M,N)$.
\end{theorem}

\begin{proof}
  Given co-wheels $G_m$ and $H_n$, add $k$ new vertices to $G_n$
and $l$ new vertices to $H_n$
(and arrows to these vertices from the roots)
and denote the resulting rooted digraphs by $G_n^k$ by $H_n^l$.
Since the new vertices are useless for both Spoiler and Duplicator, we have 
$D(G_m^k,H_n^l)=D(G_m,H_n)$ for any $k,l\ge0$.

Denote $m=M-1$ and $n=N-1$. If $m$ and $n$ are coprime, then
we can take $G=G_m$, $H=H_n$, and Lemma \ref{lem:cowheels} does the job.
Consider now the case that $m$ and $n$ are not coprime.
If $m$ is a prime divisor of $n$, then
$m$ and $n-1$ are coprime, and we can take $G=G_m$ and $H=H_{n-1}^1$ in this case.
The case that $n$ is prime is similar.
If none of $m$ and $n$ is prime, let $p<n$ be the prime closest to $n$.
By \cite{Baker.2001}, we have $p>n-n^{0.525}$ for a large enough $n$. 
Assume first that $p$ does not divide $m$. Since
these two numbers are coprime, we can take
$G=G_m$ and $H=H_p^{n-p}$ getting
$$
D(G,H)=D(G_m,H_p)>\frac12\,m(p-3)>\frac12\,m(n-n^{0.525}-3).
$$
If $p$ divides $m$, the numbers $m-1$ and $p$ are coprime, and we take
$G=G_{m-1}^1$ and $H=H_p^{n-p}$.
\end{proof}

\begin{figure}
\centering
 \begin{tikzpicture}[scale=0.5]
  \FPset{\r}{2}
\begin{scope}
  \FPset{\p}{1.464}
  \node[rvertex] (r) at (0,0) {};
  \foreach \i in {0,...,3}
   {
    \node[avertex] (x\i) at (\i*90:\r cm) {};
    \node[bvertex] (y\i) at (\i*90+30:\p cm) {};
    \node[cvertex] (z\i) at (\i*90+60:\p cm) {};
    \FPsub{\j}{\i}{1}
    \FPround{\j}{\j}{0}
    \ifnumgreater{\i}{0} 
{\draw (x\j) -- (y\j) -- (z\j) -- (x\i);}{}
    }
    \draw (x3) -- (y3) -- (z3) -- (x0);
    \draw (r) -- (x0);
\end{scope}
 
\begin{scope}[xshift=90mm]
 \FPset{\q}{1.654}
   \node[rvertex] (r) at (0,0) {};
  \foreach \i in {0,...,4}
   {
    \node[avertex] (x\i) at (\i*72:\r cm) {};
    \node[bvertex] (y\i) at (\i*72+24:\q cm) {};
    \node[cvertex] (z\i) at (\i*72+48:\q cm) {};

    \FPsub{\j}{\i}{1}
    \FPround{\j}{\j}{0}
    \ifnumgreater{\i}{0} 
{\draw (x\j) -- (y\j) -- (z\j) -- (x\i) -- (r);}{}
    }
    \draw (x4) -- (y4) -- (z4) -- (x0);
\end{scope}

 \end{tikzpicture}
\caption{\COWHEELS\  as colored graphs.}
\label{fig:cowheels-as}
\end{figure}
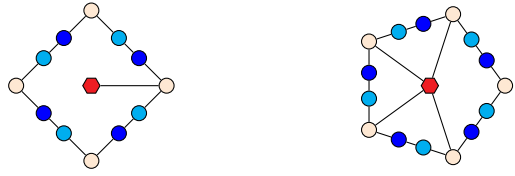

Using a simple gadget, in the \COWHEELS\  pattern we can make edges undirected
simulating directions by vertex colors.
In this way, we can construct examples of pairs with large $D(G,H)$
also for colored graphs;
see Fig.~\ref{fig:cowheels-as}.

\begin{corollary}\label{cor:colgraphs-dags}
  Theorem \ref{thm:cowheels} holds true also for colored graphs
with bound $D(G,H)\ge(\frac16-o(1))MN$.
\end{corollary}

Corollary \ref{cor:colgraphs-dags} can be
obtained also from the \DOMINO\  pattern, though with a smaller
factor $\frac18-o(1)$; see Fig.~\ref{fig:domino-colgraphs}. It is worth noting that
$G$ will be a unicyclic graph while $H$
will be a tree
 (more exactly, $H$ will be a caterpillar
and can be made even a path at he cost of further decreasing
the constant factor to $\frac1{10}-o(1)$).
Note that this result
is best possible in the sense that, by Corollary \ref{cor:A-tree},
$G$ cannot be a acyclic.

\begin{corollary}
  \label{cor:cycle-tree}
For every $M\ge2$ there is a unicyclic colored graph $G_M$ with $M$ vertices
 and for every $N\ge1$ there is a tree $H_N$ with $N$ vertices
such that $D(G_M,H_N)<\infty$ and $D(G_M,H_N)>\frac18\,(M-1)(N-5)$.
\end{corollary}

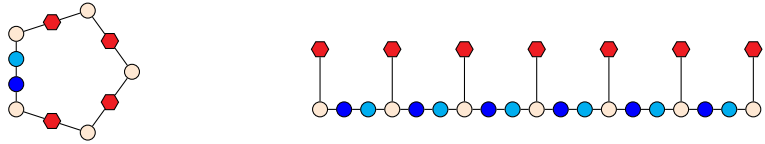
\begin{figure}
\centering
 \begin{tikzpicture}
  \FPset{\r}{2}
\begin{scope}[scale=0.5]
  \path (0,\r) node[avertex] (x0)   {}
       ++(234:1cm) node[schvertex] (y0)   {} edge (x0)
       ++(234:1cm) node[avertex] (x1)   {} edge (y0)
       ++(162:1cm) node[schvertex] (y1)   {} edge (x1)
       ++(162:1cm) node[avertex] (x2)   {} edge (y1)
       ++(90:6.7mm) node[bvertex] (y2)   {} edge (x2)
       ++(90:6.7mm) node[cvertex] (z2)   {} edge (y2)
       ++(90:6.7mm) node[avertex] (x3)   {} edge (z2)
       ++(18:1cm) node[schvertex] (y3)   {} edge (x3)
       ++(18:1cm) node[avertex] (x4)   {} edge (y3)
       ++(306:1cm) node[schvertex] (y4)   {} edge (x4) edge (x0);
\end{scope}

\begin{scope}[xshift=25mm,yshift=5mm,scale=0.8]
  \foreach \i in {0,...,6}
   {
    \node[avertex] (u\i) at (\i*12 mm, 0 cm) {};
    \node[schvertex] (s\i) at (\i*12 mm, 1 cm) {};
    \draw (u\i) -- (s\i);
    \FPsub{\j}{\i}{1}
    \FPround{\j}{\j}{0}
    \ifnumgreater{\i}{0} 
      {
    \node[bvertex] (v\j) at (\j*12 mm + 4 mm, 0 cm) {};
    \node[cvertex] (w\j) at (\j*12 mm + 8 mm, 0 cm) {};
    \draw (u\j) -- (v\j) -- (w\j) -- (u\i);
      }{}
    }
 \end{scope}
 \end{tikzpicture}
\caption{The colored graphs obtained from the \DOMINO\  example in Fig.~\protect\ref{fig:domino}.}
\label{fig:domino-colgraphs}
\end{figure}

\begin{remark}\rm
Feder and Vardi \cite{Feder.1998} showed that a general CSP is equivalent to
the homomorphism problem restricted to directed acyclic graphs (\emph{dag}s).
In view of this result, it is natural that Corollary \ref{cor:colgraphs-dags}
is true also for uncolored dags. Indeed, the directed cycles in
\COWHEELS\ can be broken by subdividing each arrow in the cycle into three
arrows oriented in different directions. The root nodes can be designated by
attaching additional arrows; see Fig.~\ref{fig:cowheels-dags}.
In fact, any distinguishable uncolored digraphs $G$ and $H$ with large $D(G,H)$ must be acyclic: 
The existence of a directed cycle or a loop in $G$ or $H$
implies that either $D(G,H)=\infty$ or $D(G,H)\le v(H)+1$. 
\comm{%
Let $G$ and $H$ be two (uncolored) digraphs. If $H$ has a loop, Duplicator
wins just by staying in it all the time.
Assume $H$ has no loop.
If $G$ has a loop, Spoiler wins in two moves by putting
both pebbles in the loop.
Assume now that both $G$ and $H$ are loopless.
If $H$ has a (directed) cycle, Duplicator wins by walking in this cycle
in one or the other direction.
Assume $H$ has no cycle.
If $G$ has a cycle, Spoiler wins, roughly, in at most $v(H)+1$ moves by
walking around the cycle always in the same direction.}
\end{remark}

\begin{figure}
\centering
\begin{tikzpicture}[>=stealth,scale=0.7]
 \FPset{\r}{2}
\begin{scope}[yshift=-50mm]
  \FPset{\p}{1.464}
  \node[uvertex] (r) at (0,0) {};
  \foreach \i in {0,...,3}
   {
    \node[uvertex] (x\i) at (\i*90:\r cm) {};
    \node[uvertex] (y\i) at (\i*90+30:\p cm) {};
    \node[uvertex] (z\i) at (\i*90+60:\p cm) {};
    \FPsub{\j}{\i}{1}
    \FPround{\j}{\j}{0}
    \ifnumgreater{\i}{0} 
{\draw [a] (x\j) -- (y\j); \draw [aa] (y\j) -- (z\j); \draw [aaa] (z\j) -- (x\i);}{}
    }
    \draw [a] (x3) -- (y3); \draw [aa] (y3) -- (z3); \draw [aaa] (z3) -- (x0);
    \draw[->] (r) -- (x0);
    \path (-.5,0) node[uvertex] (r1)   {} edge[<-] (r)
           (-1,0) node[uvertex] (r2)   {} edge[<-] (r1)
         (-1.5,0) node[uvertex] (r3)   {} edge[<-] (r2);
\end{scope}
 
\begin{scope}[xshift=70mm,yshift=-50mm]
 \FPset{\q}{1.654}
   \node[uvertex] (r) at (0,0) {};
  \foreach \i in {0,...,4}
   {
    \node[uvertex] (x\i) at (\i*72:\r cm) {};
    \node[uvertex] (y\i) at (\i*72+24:\q cm) {};
    \node[uvertex] (z\i) at (\i*72+48:\q cm) {};

    \FPsub{\j}{\i}{1}
    \FPround{\j}{\j}{0}
    \ifnumgreater{\i}{0} 
{\draw [a] (x\j) -- (y\j); \draw [aa] (y\j) -- (z\j); \draw [aaa] (z\j) -- (x\i); \draw[<-] (x\i) -- (r);
}{}
    }
    \draw [a] (x4) -- (y4); \draw [aa] (y4) -- (z4); \draw [aaa] (z4) -- (x0);
    \path (.5,0) node[uvertex] (r1)   {} edge[<-] (r)
           (1,0) node[uvertex] (r2)   {} edge[<-] (r1)
         (1.5,0) node[uvertex] (r3)   {} edge[<-] (r2);
\end{scope}
 \end{tikzpicture}
\caption{\COWHEELS\  as dags.}
\label{fig:cowheels-dags}
\end{figure}

\section{Winner proof systems}
\label{sec:proofs}

Inspired by \cite{Atserias.2004},
we now introduce a notion that allows us to define
a few useful parameters measuring the speed of constraint propagation.
In the next section it will serve as a link
between the length of the existential 2-pebble game on $(A,B)$
and the running time of an AC algorithm on input~$(A,B)$.

Let $G$ and $H$ be connected colored graphs, both with at least 2 vertices. 
A \emph{proof system of Spoiler's win on $(G,H)$} consists of
\emph{axioms}, that are pairs $(y,b)\in V(G)\times V(H)$
with $y$ and $b$ colored differently,
and derivations of pairs $(x,a)\in V(G)\times V(H)$ and a special symbol $\bot$
by the following \emph{rules}:
\begin{itemize}
\item 
$(x,a)$ is derivable from a set $\{(y,b_1),\ldots,(y,b_s)\}$ such that $y\in N(x)$
and $\{b_1,\ldots,b_s\}=N(a)$;
\item 
$\bot$ is derivable from a set $\{y\}\times V(H)$.
\end{itemize}
A \emph{proof} is a sequence $P=p_1,\ldots,p_{\ell+1}$ such that
if $i\le\ell$, then $p_i\in V(G)\times V(H)$ and it
is either an axiom or is derived from a set $\{p_{i_1},\ldots,p_{i_s}\}$
of preceding pairs $p_{i_j}$; also,
$p_{\ell+1}=\bot$ is derived from a set of preceding elements of $P$.
More precisely, we regard $P$ as a dag on $\ell+1$ nodes where a derived $p_i$ sends arrows
to each $p_{i_j}$ used in its derivation.
Moreover, we always assume that
$P$ contains a directed path from $\bot$ to each node, that is, every element of $P$
is used while deriving~$\bot$.

We define the \emph{length} and the \emph{size} of the proof $P$ as
$\len P=v(P)-1$ and $\s P=e(P)$ respectively. 
Note that $\len P$ is equal to $\ell$,
the total number of axioms and intermediate derivations in the proof.
Since it is supposed that the underlying graph of $P$ is connected, 
we have $\len P\le\s P$, where equality
is true exactly when $P$ is a tree.
The \emph{depth} of $P$ will be denoted by $\dd P$ and defined
to be the length of a longest directed path in $P$.
Obviously, $\dd P\le\len P$.

It is easy to show that a proof $P$ exists iff $D(G,H)<\infty$ 
(cf.\ part 1 of Theorem \ref{thm:proofs} below).
Given such $G$ and $H$,
define the \emph{(proof) depth} of $(G,H)$ to be the minimum
depth of a proof for Spoiler's win on $(G,H)$.
The \emph{(proof) length} and the \emph{(proof) size} of $(G,H)$
are defined similarly. We denote the three parameters by
$\dd{G,H}$, $\len{G,H}$, and $\s{G,H}$, respectively.
Note that
$
\dd{G,H}\le\len{G,H}\le\s{G,H}
$.

\begin{theorem}\label{thm:proofs}
Let $G$ and $H$ be connected colored graphs, both with at least 2 verices, 
such that $D(G,H)<\infty$.
  \begin{enumerate}
  \item 
$\dd{G,H}=D(G,H)$.
\item 
$\dd{G,H}\le\len{G,H}\le v(G)v(H)$ and this is tight up to a constant factor:
for every pair of integers $M,N\ge2$ there is a pair of colored graphs $G,H$ 
with $v(G)=M$ and $v(H)=N$ such that
$\dd{G,H}\ge(\frac1{6}-o(1))\,MN$.
\item 
$\s{G,H}<2\,v(G)e(H)+v(H)$.
\item 
For every $N$ there is a pair of colored graphs $G_N$ and $H_N$ both with $N$ vertices
such that  $\s{G_N,H_N}>\frac1{128}\,N^3$ for all large enough~$N$.
\comm{The best we can do is $\s{G_N,H_N}>(\frac2{243}-o(1))\,N^3$.}
  \end{enumerate}
\end{theorem}

Note that part 3 implies that $\s{G,H}<N^3$ if both $G$ and $H$ have $N$ vertices.
Therefore, part 4 shows that the upper bound of part 3 is tight up to a constant factor.

\begin{proof}
{\bf 1.} 
  It suffices to prove that, for every $r\ge0$, Spoiler has a strategy allowing him
to win in $r$ rounds starting from the position $(x,a)$ if and only if the pair
$(x,a)$ is derivable with depth $r$. This equivalence follows by a simple inductive
argument on~$r$.

{\bf 2.}
The upper bound follows from a simple observation that any proof can be rewritten
so that every axiom used and every derived pair apears in it exactly once.
The lower bound follows by part 1 from Corollary~\ref{cor:colgraphs-dags}.

{\bf 3.} Consider a proof $P$ where each pair $(x,a)$ appears at most once. 
Since the derivation of $(x,a)$ contributes $\deg a$ arrows in $P$,
and the derivation of $\bot$ contributes $v(H)$ arrows, we have
$$
\s P<\sum_{(x,a)}\deg a+v(H)=v(G)\sum_{a}\deg a+v(H)=2\,v(G)e(H)+v(H).
$$
The inequality is strict because there must be at least one axiom node,
which has out-degree 0.

{\bf 4.} 
Note that $\s{G,H}\ge\dd{G,H}\delta(H)$, where $\delta(H)$
denotes the minimum vertex degree of $H$. Therefore, we can take
graphs $G$ and $H$ with almost the same number of vertices and with quadratic $\dd{G,H}$, 
and make $\delta(H)$ large by adding linearly many universal vertices of a new color
 to each of the graphs. A \emph{universal} vertex is adjacent to all other vertices
in the graph. If each of the graphs receives at least two new vertices,
they make no influence on the duration of the existential 2-pebble game.

More specifically, we use the co-wheels from Lemma \ref{lem:cowheels} with
coprime parameters $m=n-1$
converted to colored graphs as in Corollary \ref{cor:colgraphs-dags};
see Fig.~\ref{fig:cowheels-as}.
Thus, we have colored graphs $G$ and $H$ with $v(G)=3n-2$ and $v(H)=3n+1$
such that $D(G,H)>\frac12\,(n-1)(n-3)$. Add green universal vertices so that
the number of vertices in each graph becomes $N=\lfloor\frac92\,n\rfloor$.
For the new graphs $G_N$ and $H_N$ we still have $D(G_N,H_N)>\frac12\,(n-1)(n-3)$
while now $\delta(H_N)\ge\frac32\,n$.
\end{proof}

\begin{remark}\label{rem:depth-length}\rm
  In general, the proof depth can be much smaller than the proof length. In fact,
for every $n$ there are two colored graphs $G$ and $H$ with $v(G)=n+1$ and $v(H)=2n$ 
such that $\dd{G,H}=2$ and $\len{G,H}=n^2$. For example, let $G$ be the star $K_{1,n}$ with
all vertices colored differently. Let the central vertex be colored in red.
In order to construct $H$, begin with the complete bipartite graph $K_{n,n}$
where one part of vertices is colored completely in red and the other part
is colored as the set of leaves in $G$. To obtain $G$, we remove a matching
($n$ pairwise non-adjacent edges) from this graph. Here we use $n+1$ colors.
This number can be made fixed similarly to \cite[Section III.F]{Berkholz.2012}.
\end{remark}

\section{Time complexity of Arc Consistency}\label{sec:time}

\subsection{Reduction to colored graphs}
\label{sec:reduction}

In this subsection we justify our focusing on colored graphs by showing
 a linear time reduction from the \GAME{} to its restriction on colored simple connected graphs 
(that also preserves the parameter~$D(A,B)$). 
The \emph{size} of a binary structure $A$ with binary relations $E_1,E_2,\ldots$
and unary relations $U_1,U_2,\ldots$ is defined to be $\e A=\sum_i|E_i^A|+\sum_j|U_j^A|$.

\begin{lemma}\label{lem:reduction}
There is a linear time reduction that takes two relational structures $A$ and $B$ with arbitrary unary and binary relations and computes two colored simple connected graphs $G$ and $H$ such that 
\begin{itemize}
\item $A$ and $B$ pass the arc consistency test iff so do $G$ and $H$,
\item $D(G,H) = \Theta(D(A,B))$,
\item $v(G) = O(\|A\|)$, $e(G)=O(\|A\|)$,
\item $v(H) = O(\|B\|)$, $e(H)=O(\|B\|)$.
\end{itemize}
\end{lemma}
\begin{proof}
For every binary relation $R$ we introduce two new vertex colors light-$R$ and dark-$R$ and replace 
every pair $(x,y)\in R$ in $A$ or $B$ by an undirected path $(x,r,r',y)$ where $r$ is colored 
light-$R$ and $r'$ is colored dark-$R$. 
Each triple $x,y,R$ is handled by its own pair $r,r'$.
Note that a loop $(x,x)\in R$ gives rise to a cycle $(x,r,r')$ of length~3.

We also have to ensure that the vertex colors in $G$ and $H$ are disjoint
even if the unary relations in $A$ and $B$ overlap. To this end,
for every unary relation $U$ and vertex $x\in U$ in $A$ or $B$
we remove $x$ from $U$ but create a new vertex $s\in U$ adjacent to~$x$.

In order to get the graphs connected, add a single vertex with a new color to both graphs 
and connect it with all other vertices.
\end{proof}

\subsection{An upper bound}
\label{sec:time-upper}

We now establish an upper bound of $O(v(G)e(H)+e(G)v(H))$ for the time complexity of the \GAME{}. 
One way to obtain this result is to use the linear-time reduction from arc consistency to the satisfiability problem for propositional Horn clauses (\Hornsat{}) presented in \cite{Kasif.1990}. The reduction transforms the input graphs $G$ and $H$ into a propositional Horn formula of size $v(G)e(H)+e(G)v(H)$ that is satisfiable iff arc consistency can be established on $G$ and $H$. The upper bound then follows by applying any linear time \Hornsat{} algorithm. Going a different way, we here show that the same bound can be achieved by a propagation-based algorithm, that we call \ACgraph. On the one hand, \ACgraph\ does much the same of what a linear time \Hornsat{} solver would do (after applying Kasif's reduction). On the other hand, it can be seen as a slightly accelerated version of the algorithm AC-4~\cite{Mohr.1986}.

\begin{theorem}
\ACgraph\ solves the \GAME{} in time $O(v(G)e(H)+e(G)v(H))$.
\end{theorem}

\begin{algorithm}[t]
\caption{\ACgraph}
\begin{algorithmic}
\State Input: Two colored connected graphs $G$ and $H$.
\State /*INITIALIZATION*/
\ForAll{$x \in V(G)$} 
	\State $D_x \gets \{a\in V(H)\mid a$ has the same color as $x\}$; 
	\If{$D_x=\emptyset$} \Return reject; \EndIf
	\EndFor
\ForAll{$x\in V(G)$, $a \in V(H)$}
	\State \texttt{counter[$x$,$a$]} $\gets |N(a)|$; 
	\If{$a \notin D_x$} 
		 add $(x,a)$ to $Q$; \EndIf
\EndFor
\State /*PROPAGATION*/
\While{$Q$ not empty}
	\State Select and remove $(x,a)$ from $Q$;
	\ForAll{$b \in N(a)$}
		\State \texttt{counter[$x$,$b$]} $\gets$ \texttt{counter[$x$,$b$]}$ - 1$;
		\If{\texttt{counter[$x$,$b$]}$=0$} 
			\ForAll{$y\in N(x)$}
				\If{$b\in D_y$}
					\State Delete $b$ from $D_y$;
					\State Add $(y,b)$ to $Q$;
					\If{$D_y=\emptyset$} \Return reject; \EndIf
				\EndIf
			\EndFor
		\EndIf
	\EndFor
\EndWhile
\State \Return accept;
\end{algorithmic}
\end{algorithm}

\begin{proof} 
We first analyze the running time. The initialization phase requires $O(v(G)v(H))$. The propagation phase takes $|N(a)|$ steps for every $(x,a)\in Q$ and $|N(x)|$ steps for every $(x,b)$ such that $\texttt{counter[$x$,$b$]}$ gets $0$. Since every pair is only put once on the queue and every counter voids out only once the total running time of the propagation phase bound by $\sum_{(x,a)\in V(G)\times V(H)} (|N(x)|+|N(a)|) = v(G)e(H)+e(G)v(H)$.

The rest is devoted to the proof of the algorithm's correctness.
Translated into the language
of the existential 2-pebble game, the problem is to decide for a given pair
of colored graphs $G$ and $H$ which of two cases  occurs:
Spoiler has a winning strategy for some number of rounds or Duplicator
has a winning strategy for any number of rounds.
Simplifying the terminology, we will say that \emph{Spoiler wins}
in the former case and \emph{Duplicator wins} in the latter case.
We begin with auxiliary notions and claims, then show that a modified version
of the algorithm is correct, and finally come back to the original version.

Given a pair $(x,a)\in V(G)\times V(H)$, we denote the existential 2-pebble game
with the initial position $(x,a)$ by $\game(x,a)$.
Suppose that $S\subsetneq V(G)\times V(H)$ is a set of pairs $(x,a)$
such that Spoiler wins the game $\game(x,a)$. We will assume that
$S$ contains all pairs of differently colored vertices.

Let $(y,b)\notin S$. Given $x\in V(G)$ and $a\in V(H)$, we call $a$
a \emph{partial $x$-certificate for} $(y,b)$ if $y\in N(x)$, $b\in N(a)$,
and $(x,a)\in S$. Assuming $y\in N(x)$, we denote the set of all partial 
$x$-certificates for $(y,b)$ by $\cert_y(x,b)$. Note that $\cert_y(x,b)=S|_x\cap N(b)$,
where $S|_x=\setdef{a\in V(H)}{(x,a)\in S}$ is the $x$-slice of $S$.
It follows that 
\begin{equation}
  \label{eq:certyy}
\cert_y(x,b)=\cert_{y'}(x,b)\text{ for any two }y,y'\in N(x),   
\end{equation}
that is, $\cert_y(x,b)$ actually does not depend on~$y$.

Furthermore, call $x$ a \emph{complete certificate for} $(y,b)$ if $\cert_y(x,b)=N(b)$.
The first of two following claims is straightforward.

\begin{clm}\label{cl:a}
If $(y,b)$ has a complete certificate, then Spoiler wins $\game(y,b)$.  
\end{clm}

\begin{clm}\label{cl:b}
  Let $G$ be connected. If Spoiler wins the existential 2-pebble game on $G$ and $H$,
then there exists a pair $(y,b)\notin S$ having a complete certificate.
\end{clm}

\begin{subproof}
Assume that no $(y,b)\notin S$ has a complete certificate and show that then
Duplicator wins the game on $G$ and $H$.

Call a vertex $x\in V(G)$ \emph{complete} if $S|_x=V(H)$.
Under the assumption made, no vertex of $G$ is complete.
Indeed, if $x$ is complete, then any adjacent to it vertex $y$
must be complete too because otherwise we would have $(y,b)\notin S$
for some $b\in V(H)$ and then $x$ would be a complete certificate
for $(y,b)$, contradictory to the assumption.
It follows by connectedness of $G$, that $S=V(G)\times V(H)$
while $S$ is supposed to be a proper subset.

The absence of complete vertices leads to the following winning
strategy for Duplicator. Assume that Spoiler pebbles a vertex $y$ in the first round.
Duplicator responds with a vertex $b$ such that $(y,b)\notin S$. Such $b$ exists
since $b$ is not complete. Let Spoiler pebble a vertex $x$ in the next round.
If $x$ and $y$ are non-adjacent, Duplicator reponds similarly (with a vertex $a$
such that $(x,a)\notin S$). If $x$ and $y$ are adjacent, then Duplicator reponds
with a vertex $a$ adjacent to $b$ such that $(x,a)\notin S$. Such $a$ exists
because otherwise $x$ would be a complete certificate for $(y,b)$.
Each subsequent round is played similarly.
\end{subproof}

We are now ready to describe an algorithm solving the existential 2-pebble game on
connected colored graphs $G$ and $H$. We will maintain a set $S\subset V(G)\times V(H)$
of pairs $(x,a)$ for which it is for sure known (certified) that Spoiler wins $\game(x,a)$.
Initially, $S$ cosists of those $(x,a)$ with $x$ and $a$ colored differently.
Our algorithm will try step by step to extend $S$. If no extention is possible any more,
the algorithm decides that Spoiler wins if $S$ reaches the full product $V(G)\times V(H)$
and that Duplicator wins if $S$ stays its proper subset.

Each time $S$ will be extended with a pair $(y,b)\notin S$ having a complete certificate.
Note that
the correctness of this procedure is ensured by Claims \ref{cl:a} and \ref{cl:b}.
By Claim \ref{cl:b}, an extension is always possible if Spoiler wins.
Thus, if Spoiler wins, the algorithm's decision will be correct because then eventually
$S=V(G)\times V(H)$. On the other hand, 
Claim \ref{cl:a} implies that $S$ consists of positions winning for Spoiler. Therefore,
if $S=V(G)\times V(H)$, then Spoiler really has a winning strategy in the game on $G$ and~$H$.

We now explain how our algorithms finds a pair $(y,b)\notin S$ with a complete certificate.
For this purpose, another set $Q\subset S$ is maintained. This set consists of 
\emph{influential} pairs $(x,a)\in S$ producing a partial $x$-certificate
for at least one pair $(y,b)\notin S$. Initially, $Q=S$ is the set of pairs of differently 
colored vertices. For each $(y,b)\notin S$ and $x\in V(G)$, we also have a counter
$\cnt_y(x,b)$ for the number of vertices $a\in N(b)$ that are still not accepted
as a partial $x$-certificate for $(y,b)$. Initially, $\cnt_y(x,b)=\deg b$.
The algorithm updates $S$ as follows. It takes an arbitrary pair $(x,a)\in Q$
and accepts $a$ as a partial $x$-certificate for all $(y,b)$ such that $y\in N(x)$
and $b\in N(a)$ by decreasing the value of $\cnt_y(x,b)$ in 1.
After this is done, the pair $(x,a)$ is not influential any more
and is removed from $Q$. Once $\cnt_y(x,b)=0$ for some $(x,b)$, this pair 
receives a complete certificate, namely $x$, and is added to both $S$ and $Q$.
This completes description of our algorithms.

The algorithm \ACgraph\ is pretty close to the slightly simplified version we just described.
Instead of $S$, \ACgraph\ maintains the set $D_x=V(H)\setminus S|_x$ for each $x\in V(G)$
and terminates as soon as $D_x=\emptyset$ for some $x$. Moreover,
the counter $\cnt_y(x,b)$ is parametrized only by $x$ and $b$,
which is justified by the equality~\refeq{certyy}.
\end{proof}

\subsection{Lower bounds}
\label{sec:time-lower}

Recall that by
a \textit{propagation-based arc consistency algorithm} we mean an algorithm that solves the \GAME{} by iteratively deleting possible assignments $a$ to a variable $x$ from the domain $D_x$ according to the arc consistency rule and rejects iff one domain gets empty. 
Let us maintain a list $L$ of deleted variable-value pairs by putting a pair $(x,a)$ there once $a$ is deleted from $D_x$. If the algorithm detects arc-inconsistency, then it is evident that $L$, prepended with axioms and appended with $\bot$, forms a proof of Spoiler's win. 
Thus, a propagation-based arc consistency algorithm can be viewed as a proof search algorithm that produces (in by-passing) a proof $P$ of Spoiler's win. This situation is related to the concept of a \emph{certifying algorithm} \cite{McConnell.2011}: Propagation-based algorithms not just detect Spoiler's win but also produce its certificate.
\hide{%
The propagation approach is guaranteed to produce one arbitrary proof of Spoiler's win, but not necessarily the shortest one. In fact, different algorithms can produce different proofs and this observation has been used to design heuristics for ordering the propagation list in arc consistency algorithms (as $Q$ in AC-graph) in order to decrease the number of propagation steps \cite{Bessiere.2006}. 
Now we can relate the time complexity of a propagation-based arc consistency algorithm to the length and size of the underlying proof it produces. In the next section we focus on parallel arc consistency algorithms and the proof depth. First, the running time of an algorithm that produces a proof $P$ is lower bounded by the number $\lenast{P}$ of derived variable-value pairs in $P$. 
}
For every derived element of $P$ an algorithm has to recognize its already derived parents.
This allows us to relate the running time to the proof size. Specifically, given 
an arbitrary propagation-based algorithm for the \GAME{}, let $\Time{G,H}$
denote the time it takes on input $(G,H)$. If the input $(G,H)$ is arc-inconsistent,
then it holds
\begin{equation}
  \label{eq:time>size}
\Time{G,H}\ge\s{G,H}.  
\end{equation}

\begin{theorem}\label{thm:time-lower}
Fix an arbitrary propagation-based algorithm.
\begin{enumerate}
\item 
Let $T_1(k,l)$ denote the worst working time of this algorithm over
colored graphs $G$ and $H$ with $e(G)=k$ and $e(H)=l$.
Then $T_1(k,l)>\frac18\,(k-1)(l-4)$ for all $k$ and~$l$. 
\item 
Let $T_2(n)$ denote the worst working time of the algorithm on
inputs $(G,H)$ with $v(G)+v(H)=n$.
Then $T_2(n)>\frac1{16}\,n^3$ for all large enough~$n$.
\end{enumerate}
\end{theorem}

\comm{%
I think that part 1 can be stated in a weaker but a bit more elegant form,
namely $t(k,l)\ge\frac1{16}\,kl$ for all $k$ and~$l$.
This costs some work because small $k$ and $l$ should then be treated
separately.}

\begin{proof}
By Corollary \ref{cor:cycle-tree}, there are colored graphs $G_k$ with $e(G_k)=v(G_k)=k$ and
$H_l$ with $e(H_l)=v(H_l)-1=l$ for which $D(G_k,H_l)$ is finite but large, specifically,
$D(G_k,H_l)>\frac18\,(k-1)(l-4)$. By the relation \refeq{time>size},
on input $(G_k,H_l)$ the algorithm takes time at least $\s{G_k,H_l}$,
for which we have $\s{G_k,H_l}\ge\dd{G_k,H_l}=D(G_k,H_l)$ by part 1 of Theorem~\ref{thm:proofs}.

Part 2 follows from part 4 of Theorem~\ref{thm:proofs}.
\end{proof}

\begin{corollary}
  In terms of the parameters $e(G)$ and $e(H)$, the time bound $O(e(G)\allowbreak\cdot e(H))$
is optimal up to a constant factor among propagation-based algorithms.
\end{corollary}

Note that
$O(e(G)v(H) + v(G)e(H))=O((v(G) + v(H))^3)$.

\begin{corollary}
  In terms of the parameter $n=v(G) + v(H)$, the time bound $O(n^3)$
is best possible for a propagation-based algorithm.
\end{corollary}

\subsection{Parallel complexity}
It is known that the \GAME{} is PTIME-complete under logspace-reductions \cite{Kasif.1990,Kolaitis.2003}. Under the assumption that PTIME $\neq$ NC, it follows that the \GAME{} cannot be parallelized. However, several parallel algorithms with a polynomial number of processors appear in the literature (e.g., \cite{Samal.1987}). 
\hide{%
This was first proven by Kasif \cite{Kasif.1990}; Kolaitis and Panattaja \cite{Kolaitis.2003} presented an alternative proof for colored graphs using the characterization as existential 2-pebble game and extended this result to higher levels of consistency. Under the assumption that PTIME $\neq$ NC, it follows that the \GAME{} cannot be parallelized. However, several parallel algorithms with a polynomial number of processors appear in the literature (e.g., \cite{Samal.1987}). 
\hide{TODO: further references}
The approach taken there is to deduce inconsistent values in parallel using the arc consistency rule. In the same fashion we can parallelize AC-graph by processing all elements of the queue $Q$ in parallel. The worst-case time complexity is then easily seen to be upper bounded by the number of successive updates of the queue which in turn is upper bounded by $O(v(G)v(H))$. Moreover, there is 
}
We are able to show
a tight connection between the running time of a parallel algorithm and the round complexity of the existential 2-pebble game. 
The following result is worth noting since $D(G,H)=\dd{G,H}$ can be much smaller that $\s{G,H}$
(cf.\ Remark \ref{rem:depth-length}),
and then a parallel propagation-based algorithm can be much faster than any sequential propagation-based algorithm.
 
\begin{theorem}
\mbox{}

  \begin{enumerate}
  \item 
\GAME{} can be solved in time $O(D(G,H))$ on a CRCW-PRAM with polynomially many processors.
\item 
Any parallel propagation-based arc consistency algorithm needs time $D(G,H)$ on arc-inconsistent instances $(G,H)$.
  \end{enumerate}
\end{theorem}

\begin{proof}
First apply the reduction from Lemma \ref{lem:reduction} to colored simple graphs in constant parallel time. 
Consider the parallelized version of AC-graph and assume that Spoiler wins the existential 2-pebble game on $G$ and $H$. The initialization phase can be implemented in constant time. For every element $(x,a)\in Q$ the propagation phase can also be processed in constant time. The algorithm iteratively processes the whole set $Q$ in parallel constant time and computes a new set $Q'$. Let $Q_i$ be the queue $Q$ after the $i$th iteration of the propagation phase and $Q_{\leq i} := \bigcup_{j \leq i} Q_j$. An easy induction shows that $Q_{\leq i}$ (for $i<D(G,H)$) is the set of all positions $(x,a)\in V(G)\times V(H)$ such that Spoiler wins in at most $i$ steps. Hence, for $l=D(G,H)$-1 it holds that $\{x\}\times V(H)\subseteq Q_{\leq l}$ where $x$ is the first vertex Spoiler puts a pebble on. But this means that in iteration $l$ all values from $D_x$ are already deleted and hence the algorithm rejects.

A parallel propagation-based algorithm produces proofs of Spoiler's win. Hence, the algorithm can derive $(x,a)$ only if the parents in the underlying proof have been derived. It follows that the parallel running time has to be lower bounded by $\dd{G,H}=D(G,H)$.
\end{proof}

\section{Conclusion and further questions}

We investigated the round complexity $D(G,H)=D^2(G,H)$ of the existential 2-pebble game on colored graphs and established lower bounds of the form $\Omega(v(G)v(H))$, which translate to lower bounds on the nested propagation steps in arc consistency algorithms. The next step in this line of research is to investigate the number of rounds $D^3(G,H)$ in the existential 3-pebble game that interacts with \emph{path consistency} algorithms in the same way as the 2-pebble game with arc consistency. Note that, similarly to
Corollary \ref{cor:upper}, $D^3(G,H) = O(v(G)^2v(H)^2)$. Ladkin and Maddux \cite{Ladkin.1994} showed that $D^3(G,H) = \Omega(v(G)^2)$ using  algebraic methods (where $H$ is a graph of constant size). Using methods from \cite{Berkholz.2012} one can construct examples of graphs with $D^3(G,H)= \Omega(v(H)^2)$ (where $G$ is a graph of constant size). It remains for future work to exhibit graphs $G_n, H_n$, both on $n$ vertices, such that $D^3(G_n,H_n) = \Omega(v(G)^2v(H)^2)$. This would translate to an $\Omega(n^4)$ lower bound for sequential and parallel path consistency algorithms. 

By Corollary \ref{cor:upper}, $D(A,B)\le v(A)v(B)+1$ for arbitrary binary structures $A$ and $B$ with $D(A,B)<\infty$.
On the other hand, for the \DOMINO\  example $A_m,B_n$ with even $n$
we have $D(A_m,B_n)>\frac12\,v(A_m)v(B_n)$. An interesting problem
is to close the gap between these bounds. We conjecture
that the lower bound of $\frac12\,v(A)v(B)$ is sharp.

Finally, we want to stress that our lower bounds for the time complexity of arc consistency hold only for constraint propagation-based algorithms. Is there a faster way to solve the \GAME{} using a different approach?


\begin{thebibliography}{10}

\bibitem{Feder.1998}
Feder, T., Vardi, M.Y.:
\newblock The computational structure of monotone monadic {SNP} and constraint
  satisfaction: A study through datalog and group theory.
\newblock SIAM Journal on Computing \textbf{28}(1) (1998)  57--104

\bibitem{Berkholz.2012}
Berkholz, C.:
\newblock Lower bounds for existential pebble games and k-consistency tests.
\newblock In: Proc. {LICS}'12. (2012)

\bibitem{Mackworth.1977}
Mackworth, A.K.:
\newblock Consistency in networks of relations.
\newblock Artificial Intelligence \textbf{8}(1) (1977)  99 -- 118

\bibitem{Bessiere.2006}
Bessi\`{e}re, C.:
\newblock Constraint Propagation.
\newblock In: Handbook of Constraint Programming. Elsevier Science Inc. (2006)

\bibitem{Lecoutre.2003}
Lecoutre, C., Boussemart, F., Hemery, F.:
\newblock Exploiting multidirectionality in coarse-grained arc consistency
  algorithms.
\newblock In: Proc. {CP}'03. (2003)  480--494

\bibitem{Dongen.2002}
Dongen, M.R.C.v.:
\newblock {AC}-3d an efficient arc-consistency algorithm with a low
  space-complexity.
\newblock In: Proc. {CP}'02. (2002)  755--760

\bibitem{Mohr.1986}
Mohr, R., Henderson, T.C.:
\newblock Arc and path consistency revisited.
\newblock Artificial Intelligence \textbf{28}(2) (1986)  225 -- 233

\bibitem{VanHentenryck.1992}
Hentenryck, P.V., Deville, Y., Teng, C.M.:
\newblock A generic arc-consistency algorithm and its specializations.
\newblock Artificial Intelligence \textbf{57} (1992)  291 -- 321

\bibitem{Bessiere.1994}
Bessi\`{e}re, C.:
\newblock Arc-consistency and arc-consistency again.
\newblock Artificial Intelligence \textbf{65}(1) (1994)  179 -- 190

\bibitem{Bessiere.1999}
Bessi\`{e}re, C., Freuder, E.C., Regin, J.C.:
\newblock Using constraint metaknowledge to reduce arc consistency computation.
\newblock Artificial Intelligence \textbf{107}(1) (1999)  125 -- 148

\bibitem{Chmeiss.1998}
Chmeiss, A., Jegou, P.:
\newblock Efficient path-consistency propagation.
\newblock International Journal on Artificial Intelligence Tools
  \textbf{07}(02) (1998)  121--142

\bibitem{Regin.2005}
R\'egin, J.C.:
\newblock {AC}-*: A configurable, generic and adaptive arc consistency
  algorithm.
\newblock In: Proc. {CP}'05. (2005)  505--519

\bibitem{Kolaitis.1995}
Kolaitis, P.G., Vardi, M.Y.:
\newblock On the expressive power of datalog: Tools and a case study.
\newblock J. Comput. Syst. Sci. \textbf{51}(1) (1995)  110--134

\bibitem{Kolaitis.2000a}
Kolaitis, P.G., Vardi, M.Y.:
\newblock A game-theoretic approach to constraint satisfaction.
\newblock In: Proc {AAAI/IAAI}'00. (2000)  175--181

\bibitem{Kolaitis.2003}
Kolaitis, P.G., Panttaja, J.:
\newblock On the complexity of existential pebble games.
\newblock In: Proc. {CSL}'03. (2003)  314--329

\bibitem{Atserias.2007}
Atserias, A., Bulatov, A.A., Dalmau, V.:
\newblock On the power of {\it k} -consistency.
\newblock In: Proc. {ICALP}'07. (2007)  279--290

\bibitem{Atserias.2008}
Atserias, A., Dalmau, V.:
\newblock A combinatorial characterization of resolution width.
\newblock J. Comput. Syst. Sci. \textbf{74}(3) (2008)  323--334

\bibitem{Berkholz.2012a}
Berkholz, C.:
\newblock On the complexity of finding narrow proofs.
\newblock In: Proc. {FOCS}'12. (2012)

\bibitem{Dechter.1985}
Dechter, R., Pearl, J.:
\newblock A problem simplification approach that generates heuristics for
  constraint-satisfaction problems.
\newblock Technical report, Cognitive Systems Laboratory, Computer Science
  Department, University of California, Los Angeles (1985)

\bibitem{Samal.1987}
Samal, A., Henderson, T.:
\newblock Parallel consistent labeling algorithms.
\newblock International Journal of Parallel Programming \textbf{16} (1987)
  341--364

\bibitem{Bessiere.2005}
Bessi\`{e}re, C., R{\'e}gin, J.C., Yap, R.H.C., Zhang, Y.:
\newblock An optimal coarse-grained arc consistency algorithm.
\newblock Artificial Intelligence \textbf{165}(2) (2005)  165--185

\bibitem{Baker.2001}
Baker, R., Harman, G., Pintz, J.:
\newblock {The difference between consecutive primes. II.}
\newblock Proc. Lond. Math. Soc., III. Ser. \textbf{83}(3) (2001)  532--562

\bibitem{Atserias.2004}
Atserias, A., Kolaitis, P., Vardi, M.:
\newblock Constraint propagation as a proof system.
\newblock In: Proc. {CP}'04. (2004)  77--91

\bibitem{Kasif.1990}
Kasif, S.:
\newblock On the parallel complexity of discrete relaxation in constraint
  satisfaction networks.
\newblock Artificial Intelligence \textbf{45}(3) (1990)  275 -- 286

\bibitem{McConnell.2011}
McConnell, R.M., Mehlhorn, K., N{\"a}her, S., Schweitzer, P.:
\newblock Certifying algorithms.
\newblock Computer Science Review \textbf{5}(2) (2011)  119--161

\bibitem{Ladkin.1994}
Ladkin, P.B., Maddux, R.D.:
\newblock On binary constraint problems.
\newblock J. ACM \textbf{41}(3) (May 1994)  435--469

\end{thebibliography}

\end{document}